\documentclass{article}
\date{}

\usepackage{amsmath,amssymb}
\usepackage{graphicx}
\usepackage[ruled,vlined,linesnumbered,commentsnumbered]{algorithm2e}
\usepackage{url}
\newtheorem{theorem}{Theorem}[section]
\newtheorem{property}{Property}[section]
\newtheorem{definition}{Definition}[section]

\newenvironment{proof}[1][Proof]{\begin{trivlist}
\item[\hskip \labelsep {\bfseries #1}]}{\end{trivlist}}

\newcommand{\qed}{\nobreak \ifvmode \relax \else
      \ifdim\lastskip<1.5em \hskip-\lastskip
      \hskip1.5em plus0em minus0.5em \fi \nobreak
      \vrule height0.75em width0.5em depth0.25em\fi}
\begin{document}

\title{d-Hop Dominating Set for Directed Graph with in-degree Bounded by One}

\author{Joydeep Banerjee, Arun Das, and Arunabha Sen \\
School of Computing, Informatics and Decision System Engineering\\
\small Arizona State University, Tempe, Arizona 85287 \\
\small Email: \{joydeep.banerjee, arun.das, asen\}@asu.edu}

\maketitle

\begin{abstract}
Efficient communication between nodes in ad-hoc networks can be established through repeated cluster formations with designated \textit{cluster-heads}. In this context minimum d-hop dominating set problem was introduced for cluster formation in ad-hoc networks and is proved to be NP-complete. Hence, an exact solution to this problem for certain subclass of graphs (representing an ad-hoc network) can be beneficial. In this short paper we perform computational complexity analysis of minimum d-hop dominating set problem for directed graphs with in-degree bounded by $1$. The optimum solution of the problem can be found polynomially by exploiting certain properties of the graph under consideration. For a digraph $G_{D}=(V_D,E_D)$ an $\mathcal{O}(|V_D|^2)$ solution is provided to the problem. 

\end{abstract}

\section{Introduction}
An Ad-hoc network is characterized by a set of dynamic nodes communicating through wireless links. Certain nodes in an ad-hoc network (namely \textit{cluster-heads} and \textit{gateways}) are elected to form a backbone \cite{baker1981architectural}  which support efficient inter-communication in between nodes. All the messages generated by nodes inside a cluster is routed via a cluster-head. Gateway nodes function as an entity for inter-cluster communication. Pertaining to the dynamics of the nodes in the network, certain nodes in a cluster might go out of reach from its cluster-head. This necessitates continuous re-election of cluster-heads.

In defining a cluster, the ad-hoc networks needs to be represented in form of a digraph $G_D=(V_D,E_D)$. The vertex set $V_D$ consists of all nodes in the network. The edge set $E_D$ consist of ordered pairs. An edge $(v_1,v_2)$ exists if a node $v_2$ is in the wireless range of a node $v_1$. A node $v_2$ is d-hops away from node $v_1$ if there exists a path $v_1 \rightarrow v_{p1} \rightarrow ... \rightarrow v_{p(d-1)} \rightarrow v_2$ in the graph $G_D$ with $d$ being the smallest such integer. In \cite{amis2000max}, a cluster formation problem is proposed in which a node belonging to a cluster is either a cluster-head or is at most d-hops away from the cluster. This is referred to as the minimum \textit{d-hop dominating set problem }. The problem is proved to be NP complete. In \cite{amis2000max} only bidirectional links between two nodes (i.e. a network represented by an undirected graph) is considered. In this paper the problem is redefined for digraphs to achieve the scope of having both unidirectional and bidirectional links. For a bidirectional link between nodes $v_1$ and $v_2$ a pair of directed edges $(v_1,v_2)$ and $(v_2,v_1)$ is included in the edge set $E_D$ of the graph $G_D$. With these definitions the decision version of minimum d-hop dominating set problem for digraphs is stated as follows: \\ \\
\textbf{Instance: }A directed graph $G_D=(V_D,E_D)$, two positive integers d and K. \\ \\
\textbf{Question: }Is there a subset $|V_D^{'}| \subseteq V_D$ with $|V_D^{'}| \le K$ such that every vertex $v \notin V_D-V_D^{'}$ is at most $d$ hops away from at least one vertex in $V_D^{'}$. 

In a directed graph $G_D=(V_D,E_D)$ a node $u$ dominates a node $v$ if the edge $(u,v) \in E$. For directed graphs, dominating set is proved to NP complete \cite{albers2004algorithms}. But for directed graphs with in-degree of at most $1$, the problem can be solved polynomially \cite{albers2004algorithms}. With a similar context, in this short paper we analyze the minimum \textit{d-hop dominating set problem} with a restricted sample space. The restriction is imposed on the graph $G_D$ which has its in-degree bounded by $1$. We show the existence of a polynomial time algorithm of the restricted version for the problem by exploiting certain properties of the graph under consideration. This is beneficial for attaining exact solution to all cluster-head re-election phases where the graph has the given property. The solution also has application to approach a subclass of particular problem in interdependent network \cite{Das}. To best of our knowledge, all subsequent works which define a cluster as in \cite{amis2000max} or any other analysis of d-hop dominating set problem doesn't approach this restricted version. 

\section{Analysis of Restricted Version of the Problem}
Certain properties and definitions of a directed graph with in-degree bounded by one (denoted as the graph $G_{D}=(V_{D},E_{D}$)) are stated before the analysis of the problem. 

\begin{property}
A weakly connected subgraph $G_{SD}$ of the directed graph $G_{D}$ has the following properties --- 1) $G_{SD}$ can't have more than one cycle (shown in Figure 1(a)). If there exists more than one cycle then at least one node has an in-degree greater than 1. 2) Owing to (1) no cycle of $G_{SD}$ can have a subgraph whose vertices and edges form another cycle (shown in Figure 1(b)). 3) Additionally all nodes of a cycle has no directed incoming edge from another node not in the cycle (shown in Figure 1(c)). 
\end{property}
 Hence based on \textit{Property 2.1}, the subgraph $G_{SD}$ is either a Directed Acyclic Graph (i.e. a DAG with in-degree of all nodes bounded by 1) or has one cylce with all directed edges going out from nodes in the cycle.

\begin{figure}[ht]
  \centering
    \includegraphics[width=0.6\textwidth]{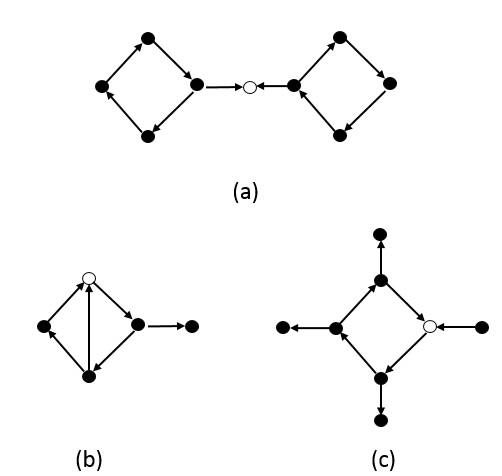}
  \caption{Figures showing weakly connected graphs (a) with two cycles, (b) with a cycle and a subgraph of the cycle which also forming a cycle, (c) cycle having a node which has an incoming edge. All the graphs have at least one node with in-degree greater than one (marked as white) }
\label{fig:fig1}
\end{figure}

\begin{property}
The number of directed edges of the graph $G_{D}$ is upper bounded by $|V_D|$ as all nodes have in-degree bounded by $1$.
\end{property}

\begin{property}
 If the subgraph $G_{SD}$ (as in \textit{Property 2.1}) is a DAG then there exists exactly one node which has an in-degree of zero (otherwise there would be a cycle). Such a node is referred as the \textbf{root node}.  
\end{property}

\begin{definition}
A \textbf{leaf node of a weakly connected component} is defined as the node with an incoming edge and no outgoing edge. 
\end{definition}

\begin{definition}
The \textbf{distance} of a vertex $v$ in a weakly connected component is defined as the number of hops by which $v$ is away from the root node if the graph is a DAG. If the weakly connected component has a cycle then the distance of vertex $v$ is defined as the number of hops by which $v$ is away from the closest node in the cycle. All vertices in the cycle has a distance value of $0$.
\end{definition}

\begin{definition}
An \textbf{isolated strongly connected component} is defined as the component which is not a subgraph of any weakly connected component of the graph $G_{D}$ 
\end{definition}

\begin{property}
Each isolated strongly connected subgraph of the directed graph $G_{D}$ is just a single cycle(follows directly from \textit{Property 2.1}). 
\end{property}

\begin{property}
The graph $G_D$ consists of weakly connected components and isolated strongly connected components as given by \textit{Property 2.1} and \textit{2.4} respectively. 
\end{property}

Exploiting the properties of the graph $G_{D}$ an algorithm (Algorithm 1) is designed to solve the restricted minimum \textit{d-hop dominating set problem for directed graphs}. The proof of optimality and time complexity analysis of the algorithm are provided in Theorem 2.1 and Theorem 2.2 respectively. \\
\begin{algorithm}[ht]
\small
	\KwData{A directed graph $G_D=(V_D,E_D)$ with in-degree bounded by one and a positive integer $d$
		}
	\KwResult{minimum d-hop dominating set $D$ for the graph $G_D$
		}
	\Begin{			
			Set $D=\emptyset$\;		
			Compute all weakly connected components and isolated strongly connected components of the graph $G_{D}$\;
			\For  {(Each weakly connected component $G_{W}=(V_{W},E_{W})$ of graph $G_{D}$)}{
					Set of covered vertices $S = \emptyset$ \;
     				\While{ ($S \ne V_{W}$) } {
      					Pick the node $u \in V_{W} / S$ in $G_{W}$ having the highest distance. If there exists more than one node then pick arbitrarily\; \vspace{0.03in}
      					\If{If the distance value of node $u$ is $0$ and $V_{W} / S \ne \emptyset$}{
      						Form individual graphs $G' = (V',E')$ with $V' \ in V_{W} / S $ composed of each connected component \;
      						Merge two connected components in vertices in $V_{W}$ can be used to merge them and the merging uses $ < d$ vertices \;
      						Select vertices from the merged components and add them to $D$ such that all vertices in $V_{W} / S$ are d-hop dominated and the number of selected vertices is minimized \;
      						Break \;
      					}
    					Include node $u$ in set $D$ such that the number of hops from $u$ to $v$ is
 				           maximum but is less than or equal to $d$\;\vspace{0.03in}
   						Update set $S$ by including vertex $u$ along with all other vertices that are $d$ hop dominated by $v$\;
					}
			}
   \For  {(Each isolated strongly connected component $G_{S}=(V_{S},E_{S})$ of graph $G_{D}$)} {
     Include $\lceil\frac{|V_S|}{d+1}\rceil$ nodes in set $D$ with number of nodes between each vertex picked being $\ge 1$ and $\le d$ \;
    }
	}		
\caption{Algorithm for finding minimum d-hop dominating set of graph $G_{D}$}
\end{algorithm}
\begin{theorem}
Algorithm 1 gives the optimum solution of d-hop dominating set problem for the graph $G_{D}$.
\end{theorem}

\begin{proof}
For an isolated strongly connected component $G_{S}=(V_{S},E_{S})$ at least $\lceil\frac{|V_{S}|}{d + 1}\rceil$ nodes has to be included in the solution. Algorithm 1 selects $\lceil\frac{|V_{S}|}{d + 1}\rceil$ nodes for each strongly connected component with number of nodes between each vertex picked being $\ge 1$ and $\le d$. Thus it includes the optimum number nodes in the solution for all strongly connected components. For a weakly connected component all the vertices included in the solution with respect to the node having highest distance value $>0$ is optimal. This can be proved by contradiction. If a node is included in the solution that does not $d-hop$ dominate the node having current highest distance value at any iteration then another vertex needs to be included in the solution to dominate it. Hence the cardinality of the solution set would increase. After a certain number of iterations in the \textit{while} loop (line $6 - 13$) the nodes that are not $d$-hop dominated (if exist) would essentially be the nodes in the cycle. So they would have a distance value of $0$ and the algorithm would enter the computation steps as in line $8-11$. Each graph formed in line $9$ would essentially be a path graph. Two graphs $G'_1$ and $G'_2$ are merged using vertices inside the cycle and \textit{iff} the number of vertices to include is $<d$. This ensures when a vertex in $G'$ is selected in solution it would take into account the vertices it can dominate in $G''$. Selecting the minimum number of vertices that finds the $d$-hop dominating set of these merged components is straightforward and is easily seen to be optimal. Hence Algorithm 1 returns an optimal solution to the $d$-hop dominating set problem with in-degree bounded by $1$.
\end{proof}
\begin{theorem}
Algorithm 1 solves d-hop dominating set problem for the graph $G_{D}$ polynomially with time complexity of order $\mathcal{O}(|V_D|^{2})$ 
\end{theorem}

\begin{proof}
All strongly connected components of the graph $G_D$ can be found using Tarjan's algorithm \cite{tarjan} in $\mathcal{O}(|V_{D}|+|E_{D}|) = \mathcal{O}(|V_{D}|)$ (as $|E_D| \le |V_D|$). The isolated strongly connected components can be separated by checking the out-degree of all nodes in a strongly connected components which should be exactly equal to $1$. This is done in $\mathcal{O}(|V_{D}|)$. All other components form the weakly connected components. Hence step 3 takes $\mathcal{O}(|V_{D}|)$. The first \textit{for} loop computes the minimum d-hop dominating set for all weakly connected components. In worst case all the vertices of the graph $G_D$ belongs to some weakly connected components. Consider that there are $m$ weakly connected components (with $m \le |V_D|$) and the graphs $G_1=(V_1,E_1)$, ..., $G_m=(V_m,E_m)$ represent the components. Accordingly, the \textit{for} loop in step 4 iterates for $m$ times. In each iteration of the \textit{while} loop in step 6, at least $1$ nodes is included in $S$. Hence, in $i^{th}$ iteration of the \textit{for} loop, the while loop iterates for at most $|V_i|$ times. For a given weakly connected graph $G_x=(V_x,E_x)$ in the $x^{th}$ iteration of the \textit{for}, the distance of the nodes (in step 6) can be found in $\mathcal{O}(|V_x|+|E_x|)=\mathcal{O}(|V_x|)$ (considering the fact $|E_x| \le |V_x|$ and the distance values for nodes in any component has to be computed exactly once). The computations in step $13-14$ can then be done in $\mathcal{O}(|V_x|)$ . Similarly the computations in the branch of lines $9-11$ can be done in $\mathcal{O}(|V_x|)$. So the running time of the while loop is $\mathcal{O}(|V_x|^{2})$ considering the while loop iterates for the maximum number of times. In overall the running time of the \textit{for} loop is bounded by $max_{m} (\lceil \sum_{i = 1}^{m} |V_{D_m}|^2 \rceil) = \mathcal{O}(|V_D|^2)$ where $V_{D_m}$ are vertices in the weakly connected component at $m^{th}$ iteration.

The second \textit{for} loop starting in step 15 computes the minimum d-hop dominating set for all isolated strongly connected component. Again, in the worst case all vertices of the graph $G_D$ belongs to some isolated strongly connected component. Similarly consider that there are $m$ components (with $m \le \frac{|V_D|}{2}$) and the graphs $G_1=(V_1,E_1)$, ..., $G_m=(V_m,E_m)$ represent the components. For an iteration $i$ of the \textit{for} loop the computation time of step 19 is bounded from above by $|V_i|$. Hence with $m$ iteration of the for loop the computation complexity in steps $18-19$ is upper bounded by $\overset{m}{\underset{i=1}{\sum}}|V_i| = \mathcal{O}(|V_{D}|)$. So the total time complexity for finding minimum d-hop dominating set for isolated strongly connected components as in step $18-19$ is $\mathcal{O}(|V_{D}|)$. Hence Algorithm 1 takes in total an $\mathcal{O}(|V_{D}|^2)$ time to compute the optimum solution of the problem. 
\end{proof}

\section{Conclusion}
In this short paper we analyze the minimum d-hop dominating set problem for directed graph with in-degree bounded by one. It is found that exploiting certain properties of the graph under consideration an algorithm can solve the problem in polynomial time, with run time complexity bounded by two times the number of vertices in the graph. This result can be used in any application (as in \cite{Das}), apart from cluster-head election in wireless ad-hoc networks, where 1) the problem can be formulated as d-hop dominating set 2) the underlying graph has in-degree of all nodes bounded by one.


\begin{thebibliography}{4}
\bibitem{baker1981architectural}Baker, Dennis J., and Anthony Ephremides. "The architectural organization of a mobile radio network via a distributed algorithm." Communications, IEEE Transactions on 29.11 (1981): 1694-1701.
\bibitem{amis2000max}Amis, Alan D., et al. "Max-min d-cluster formation in wireless ad hoc networks." INFOCOM 2000. Nineteenth Annual Joint Conference of the IEEE Computer and Communications Societies. Proceedings. IEEE. Vol. 1. IEEE, 2000.
\bibitem{albers2004algorithms}Chlebík, Miroslav, and Janka Chlebíkova. "Approximation hardness of dominating set problems." Algorithms–ESA 2004. Springer Berlin Heidelberg, 2004. 192-203..
\bibitem{tarjan} Tarjan, R. E. "Depth-first search and linear graph algorithms", SIAM Journal on Computing 1 (2): 146–160, 1972.
\bibitem{Das} Das, Arun, Joydeep Banerjee, and Arunabha Sen. "Root cause analysis of failures in interdependent power-communication networks." Military Communications Conference (MILCOM), 2014 IEEE. IEEE, 2014.

\end{thebibliography}
\end{document}